\documentclass[12pt]{article} 

\usepackage{times}
\usepackage{amssymb}
\usepackage{amsmath}
\usepackage{graphicx}
\usepackage{epsfig}
\usepackage{subfigure}

\def\shadowbox{\hbox{\rule[-0.0ex]{0.1ex}{1.2ex}%
\hspace{-0.1ex}\rule[-0.0ex]{1.2ex}{0.1ex}%
\hspace{0.0ex}\rule[-0.0ex]{0.1ex}{1.2ex}\hspace{-1.3ex}%
\rule[1.15ex]{1.25ex}{0.1ex}\hspace{-0.0ex}\rule[-0.25ex]{0.3ex}{1.1ex}%
\hspace{-1.2ex}\rule[-0.25ex]{1.1ex}{0.25ex}}}
\def\qed{\ifmmode \hbox{\hfill\shadowbox}
     \else \hphantom{x}\hfill\shadowbox \fi}

\newtheorem{theorem}{Theorem}[section]
\newtheorem{lemma}[theorem]{Lemma}

\newtheorem{proposition}[theorem]{Proposition}

\newtheorem{Signal Set}[theorem]{Signal Set}

\def\R {\mathbb{R}}
\def\E {\mathbb{E}}
\def\Z {\mathbb{Z}}
\def\S {\mathcal{S}}

\def\P {\mathbb{P}}

\def\YL {Y^{(L)}}

\def\ks {k^{*}}

\allowdisplaybreaks[3]

\begin{document}

\title{Expected Supremum of a Random Linear Combination of Shifted Kernels}
%~\IEEEmembership{Fellow,~IEEE,}
\author{Holger Boche, Brendan~Farrell,
        Michel Ledoux and Moritz Wiese
\thanks{\noindent H. Boche and M. Wiese 
are with the Lehrstuhl f\"ur Theoretische Informationstechnik, 
Technische Universit\"at M\"unchen, Arcisstr. 21, 80333 M\"unchen, Germany. 
B. Farrell was with the Lehrstuhl f\"ur Theoretische Informationstechnik when this 
work was completed. 
He is now with the Department of Computing and Mathematical Sciences, 
California Institute of Technology, 1200 E. California Blvd., Pasadena, CA 91125, USA. 
M. Ledoux is with the Institut de Math\'ematiques de Toulouse, Universit\'e de Toulouse, 
31062 Toulouse, France and Institut Universitaire de France.
\newline
E-mail: boche,wiese@tum.de,\;farrell@cms.caltech.edu,\;ledoux@math.univ-toulouse.fr
\newline 
H. Boche was supported by start-up funds of the Technische Universit\"at M\"unchen. 
}}
\date{}

\maketitle

\begin{abstract}
We address the expected supremum of a 
linear combination of shifts of the sinc kernel with random coefficients. 
When the coefficients are Gaussian, the expected supremum is of order $\sqrt{\log n}$, 
where $n$ is the number of shifts. 
When the coefficients are uniformly bounded, the expected supremum is of order $\log\log n$. 
This is a noteworthy difference to orthonormal functions on the unit interval, where the 
expected supremum is of order $\sqrt{n\log n}$ 
for all reasonable coefficient statistics.  
\end{abstract}


$\;$ 

{\bf AMS Classification Numbers: } 60G70, 42A61, 94A12

\section{Introduction}

Perhaps the most fundamental functions in signal processing are shifts of the sinc kernel $\frac{\sin \pi t}{\pi t}$. 
This kernel decays slowly in time, and consequently it is generally not used in practice. 
Nonetheless, it is the starting point, certainly historically, for much of signal processing, information theory and sampling theory. 
If each shifted kernel has a random coefficient, it is natural to investigate the 
properties of the resulting signal. 
Here we address the expected supremum of such a signal. 
We let $\{a_k\}_{k=1}^\infty$ be independent random variables and consider the quantity 
\begin{equation}
\S_n=\sup_{t\in \R}\left|\sum_{k=1}^n a_k \frac{\sin \pi (t-k)}{\pi (t-k)}\right|.\label{defS}
\end{equation}
We investigate the behavior of the peak when the  $\{a_k\}_{k=1}^\infty$ are Gaussian 
and symmetric $\pm 1$ random variables, 
and show that 
\emph{in the Gaussian case $\E \S_n\sim \sqrt{\log n}$ while in the $\pm 1$ case $\E  \S_n\sim \log \log n$}. 
This result is fundamental enough to be relevant in numerous settings. 
One example is when coefficients are quantized and the $\{a_k\}_{k=1}^\infty$ 
represent the difference between an actual coefficient and its quantized value. 
Another is when the coefficients are viewed as carrying information, and one is concerned 
with the peak value of the signal. 
We discuss this briefly below.

\section{Problem Formulation and Main Result}\label{result}

We compare a signal of the type given in~\eqref{defS} with linear combinations of orthonormal functions on the unit interval. 
Here, the fundamental theorem, due to 
Kashin and Tzafriri~\cite{KT95}, 
states 
that if $\{\phi_k\}_{k=1}^\infty$ are uniformly bounded (i.e. in $\|\cdot\|_\infty$) orthonormal functions on $[0,1]$ and 
$\{a_k\}_{k=1}^\infty$ are independent symmetric random variables with a uniform bound on the third moment, then 
\begin{equation}
\E \sup_{t\in[0,1]}\left| \sum_{k=1}^n a_k\phi_k(t)\right|\sim \sqrt{n\log n}.
\end{equation}
Thus, for uniformly bounded functions on the unit interval, the necessary linear combinations 
occur and result in  
Gauss-like behavior. 
Consequently, the statement is not sensitive to the distribution of the individual coefficients.  
Note, though, that the uniform bound on the functions  $\{\phi_k\}_{k=1}^\infty$ is essential to the result. 

The theorem just stated applies to two systems of practical importance, namely the Fourier and 
Walsh systems, which are known in electrical engineering as OFDM and CDMA systems. 
Here the motivation for 
understanding the behavior of a signal's peak is that amplifiers particularly 
distort or eliminate the peak. 
This has led to extensive research in communications engineering on what is called the 
\emph{peak-to-average power ratio}. 
See the book~\cite{Lit07} for an overview of this area for OFDM (Fourier) 
and~\cite{BF11a} for  
recent work on the CDMA (Walsh) case.

Here we address similar questions for the shifted sinc kernel on the real line. 
We make several introductory observations about the equation~\eqref{defS} before formally stating the problem. 
Note that when $t=l$, $l$ an integer, the value of the function inside the absolute value bars equals $a_l$,  
so that  $\max_{1\leq k\leq n}|a_k|$ is an a priori lower bound on $\S_n$. 
Thus, a first point of interest is to compare 
the signal's peak off the set of integers to that at the integers. 
A second point is to compare the peak behavior when the random coefficients are $\pm 1$ random variables and when they are  
Gaussian variables. 
For example, the simple lower bound $\max_{1\leq k\leq n}|a_k|$ does not grow in $n$ for $\pm 1$ random variables. 
(Uniformly bounded, zero-mean random variables will be shown to behave the same as random $\pm 1$, and so we discuss only the latter at this point.) 
In such a linear combination one has sums of other independent random variables, 
yet it is unclear a priori if they behave close to Gaussian random variables.

We clarify the  dichotomy between orthonormal functions on the unit interval (multi-carrier systems in communications) and 
shifted kernels on the real line (single-carrier systems). 
In the former, the expected peak value behaves like $\sqrt{n\log n}$ as long as the individual distributions satisfy a third moment condition. 
In the latter case, the behavior depends on the individual distributions. 
The linear combination of $\pm 1$'s does \emph{not} behave like Gaussian random variables and, in particular, the expected 
value of the supremum is significantly smaller in the $\pm 1$ case. 
%In the Gaussian case $\E \S_n\sim \sqrt{\log n}$ and in the $\pm 1$ case $\E S_n\sim \log \log n$. 
The behavior in the Gaussian case follows from well-known theorems due to Slepian and Sudakov, and so the contribution here is the $\pm 1$ case.

Recall the definition of $\S_n$ from equation~\eqref{defS}. 
Our main theorem is the following. 
\begin{theorem}\label{maintheorem}
Assume the random variables $\{a_k\}$ are independent and distributed according to $\mathcal{N}(0,\sigma^2)$. 
Then there exists a constant $c_1>0$ independent of $\sigma$ such that 
\begin{equation*}
c_1 \sigma\sqrt{\log n}\leq \E\; \S_n\leq c^{-1}_1\sigma\sqrt{\log n}
\end{equation*}
for all large $n$. 
If $\{a_k\}_{k=1}^\infty$ have symmetric distribution  
and there exist constants $M$ and $m$ such that 
$\P(|a_k|>M)=0$ and $\E|a_k|\geq m>0$ for all $k$, then 
there exists a constant $c_2>0$ independent of $M$ and $m$ such that 
\begin{equation*}
c_2 m\log\log n\leq \E\; \S_n\leq c_2^{-1}M\log\log n
\end{equation*}
for all large $n$. 
\end{theorem}

Before turning to the proof, we briefly highlight how the result is tied to the non-integrability of the sinc kernel. 
If a kernel is unbounded, then a linear combination of shifts of the kernel will generally be unbounded, and so we may consider only 
bounded kernels. 
If the kernel $s$ is bounded and integrable, then one has 
\begin{equation*}
 \sup_{t\in[0,1)} \sum_{k\in\Z}|s(k+t)|<C.
\end{equation*}
Therefore, 
\begin{eqnarray*}
\sup_{t\in\R} \Big|\sum_{k=1}^n a_k s(t-k)\Big|&\leq & \max_{1\leq k\leq n}|a_k|\sup_{t\in[0,1)}\sum_{k\in\Z} |s(t-k)|\\
&\leq & C \max_{1\leq k\leq n}|a_k|.
\end{eqnarray*}
Therefore, if the random variables $\{a_k\}_{k=1}^\infty$ are uniformly bounded, 
a linear combination of the form $\sum_{k=1}^n a_k s(t-k)$ is also uniformly bounded. 
Thus, the statement in Theorem~\ref{maintheorem} is a consequence, as one expects, of the non-integrability of the sinc kernel.

\section{Proof of Main Result}\label{sectionproof}

As commented earlier, the Gaussian case follows from theorems of Slepian and Sudakov. 
For the $\pm 1$ case, we first reduce the problem to determining the expected maximum over a finite set. 
Working with this finite set will be the majority of the paper. 
We first prove a proposition that is unencumbered by several details that are necessary for the full proof. 
We do this to emphasize the aspect of the proof that is most important, namely the statement of the proposition.  
Additionally, we think that the proposition could quite likely be useful elsewhere. 
The proof of the main theorem then brings the original problem statement to the form addressed by the proposition.

%or the compact support case, we isolate the most important aspect and address it in the following proposition. 
%e also state this proposition because it may be useful in other contexts. 
%he proof of Theorem~\ref{maintheorem} will then consist in relating our signal to the proposition. 

\begin{proposition}\label{propdiscrete}
Assume $\{a_k\}_{k=1}^\infty$ are independent, symmetric and satisfy $\P(|a_k|>M)=0$ for some $M<\infty$ and $\E|a_k|\geq m>0$ 
for all $k$. 
For $1\leq k\leq n$ let 
\begin{equation*}
X_k =\sum_{l=1}^k\frac{1}{l}a_{k-l+1}+\sum_{l=2}^{n-k+1}\frac{1}{l}a_{l+k-1}.
\end{equation*}
Then there exists a constant 
$c>0$ independent of $M$ and $m$ such that for all large $n$
\begin{equation*}
cm\log\log n\leq \E \Big(\max_{1\leq k\leq n} |X_k|\Big) \leq c^{-1}M\log\log n.
\end{equation*}
\end{proposition}

\begin{proof} 
We assume that $n$ is large enough for several simple inequalities to hold. 
For $1\leq k\leq n$ we set 
\begin{equation}
Y_k=\sum_{l=1}^k\frac{1}{l}a_{k-l+1}\label{defY}
\end{equation}
and 
\begin{equation}
Z_k=\sum_{l=2}^{n-k+1}\frac{1}{l}a_{l+k-1}.\label{defZ}
\end{equation}
We have 
\begin{eqnarray*}
\E \Big(\max_{1\leq k\leq n} |X_k|\Big) &\leq & \E \Big(\max_{1\leq k\leq n} |Y_k|\Big) +\E\Big(\max_{1\leq k\leq n} |Z_k|\Big).
\end{eqnarray*}
Since the $\{a_k\}_{k=1}^\infty$ are not required to be identically distributed, 
$\max_{1\leq k\leq n} |Y_k|$ and $\max_{1\leq k\leq n} |Z_k| $ are not necessarily identically distributed. 
However, the same technique can be used to bound the expectation of both terms, 
and so we give the argument for the $Y$ term and then apply it to both. 
 
Now let $L$ be a number $2\leq L\leq n$ to be chosen later. 
If $k\leq L$, then 
\begin{equation}
|Y_k|\leq M\sum_{l=1}^k\frac{1}{l}\leq M+M\log(k)\leq M(1+\log L).\label{useL}
\end{equation}
If $k>L$ define 
\begin{equation}
\YL_k=\sum_{l=L+1}^k\frac{1}{l}a_{k-l+1}.\label{defYL}
\end{equation}
Therefore 
\begin{eqnarray*}
\E \Big(\max_{1\leq k\leq n} |Y_k|\Big)&\leq & \E \Big(\max_{1\leq k\leq L} |Y_k|+\max_{L<k\leq n}|Y_k|\Big)\\
&\leq & \E \Big(\max_{1\leq k\leq L}|Y_k|+ \max_{L<k\leq n}\Big|\sum_{l=1}^L\frac{1}{l}a_{k-l+1}+ \YL_k\Big|\Big)\\
&\leq& 2M(1+\log L)+\E \Big(\max_{L<k\leq n}| \YL_k|\Big).
\end{eqnarray*}
Using $\sum_{l=L+1}^\infty \frac{1}{l^2}\leq \frac{1}{L}$, Hoeffding's inequality, Theorem~{2} in~\cite{Hoe63}, gives
\begin{equation*}
\P\Big(\Big|\sum_{l=L+1}^k\frac{1}{l}a_{k-l+1} \Big|>t\Big)\leq 2e^{-Lt^2/2M^2}\label{asum}
\end{equation*}
for every $t>0$ and for every $k=L+1,\ldots,n$.
Then, for any $\delta>0$
\begin{eqnarray}
\E \Big(\max_{L<k\leq n}\Big|\YL_k\Big|\Big)&=& \int_0^\infty \P\Big(\max_{L<k\leq n}|\YL_k|>t\Big)dt\nonumber\\
&\leq &\delta+\int_\delta^\infty \P\Big(\max_{L<k\leq n}|\YL_k|>t\big)dt\label{union-n}\\
&\leq & \delta+ 2n\int_\delta^\infty e^{-Lt^2/2M^2}dt\label{factor-n}\\
&=&\delta +\frac{n\sqrt{2\pi}M}{\sqrt{L}}e^{-\delta^2 L/2M^2}.\label{abc}
\end{eqnarray}
Setting $\delta=M\frac{\sqrt{2\log n}}{\sqrt{L}}$, 
we have
\begin{eqnarray*}
\eqref{abc}&\leq& M\frac{\sqrt{2\log n}}{\sqrt{L}}+M\frac{\sqrt{2\pi}}{\sqrt{L}}.
\end{eqnarray*}
Therefore 
\begin{eqnarray*}
\E \Big(\max_{1\leq k\leq n} |Y_k|\Big)&\leq & 2M(1+\log L) + M \frac{\sqrt{2\log n}}{\sqrt{L}}+M\frac{\sqrt{2\pi }}{\sqrt{L}} .
\end{eqnarray*}
Setting $L=\log n$, we have
\begin{eqnarray*}
\E \Big(\max_{1\leq k\leq n} |Y_k|\Big)&\leq & 2M(2+\log \log n) + M\frac{\sqrt{2\pi}}{\sqrt{\log n}},
\end{eqnarray*}
and for $n$ large enough, applying the same argument to $Z_k$,  
\begin{eqnarray*}
\E \Big(\max_{1\leq k\leq n} |X_k|\Big)&\leq & 4M\log \log n .
\end{eqnarray*}

Now we prove the lower bound. 
Let $\{\epsilon_k\}_{k=1}^\infty$ be independent symmetric $\pm 1$ random variables, so that $\{\epsilon_k |a_k|\}_{k=1}^\infty$ has the same distribution 
as $\{a_k\}_{k=1}^\infty$. 
Let $[\cdot]$ denote the integer part of a positive real number. 
For a constant $c_1>0$, consider the subintervals 
$\{j[c_1\log_2 n]+1,\ldots,(j+1)[c_1\log_2 n]\}$ for $j=0,\ldots,[n/[c_1\log_2 n]]-1$. 
Let $J$ denote the number of such intervals, i.e. $J=[n/[c_1\log_2 n]]$. 
In \cite{ER70} it is shown in equation (2.14) that there exist constants $0<c_1,c_2,c_3$ such that with probability at least $1-e^{-n^{c_3}}$, 
the sum of the r.v. $\{\epsilon_i\}$ corresponding to at least one of the subintervals just described of length $[c_1\log_2 n]$ is at least $c_2\log_2 n$. 

For completeness we sketch the argument given in~\cite{ER70}. 
%First, note that if the sum of $\pm 1$'s in an interval of length $[c_1\log_2 n]$ is at least $c_2\log_2 n$, then the number of $+1$'s 
%is at least $\frac{1+c_2}{2}c_1\log_2 n$. 
Setting $\gamma=\frac{1+c_2}{2}$ and $K= c_1 \log_2 n$,
\begin{eqnarray}
\lefteqn{\P(\textnormal{number of +1's in each interval is less than}\;\gamma c_1\log_2 n )}\nonumber\\
&\leq & (\P(\textnormal{number of +1's in one interval is less than}\;\gamma c_1\log_2 n ))^{\frac{n}{K}-1}\nonumber\\
&=& (1-\P(\textnormal{number of +1's in one interval is  greater than}\;\gamma c_1\log_2 n ))^{\frac{n}{K}-1}\nonumber\\
&=&\Bigg( 1-2^{- K}  \sum_{\gamma K\leq l\leq K}   \binom{K}{l} \Bigg)^{\frac{n}{K}-1}\nonumber\\
&\leq & (1- A_1   K^{-1/2} 2^{K(h( \frac{1+c_2}{2})-1)})^{\frac{n}{K}-1}\label{usesterling},
\end{eqnarray}
where in~\eqref{usesterling} we have used Stirling's formula for $\frac{1}{2}<\gamma<1$, where for $0<x<1$, 
$h$ is the entropy $h(x)=x\log_2(\frac{1}{x})+(1-x)\log_2(\frac{1}{1-x})$. 
By choosing $c_2$ small enough, 
$h( \frac{1+c_2}{2})-1=\frac{-1}{c_1}+\frac{2\delta_2}{c_1} <0$ for some $0<\delta_2<\frac{1}{2}$. 
Then 
\begin{eqnarray*}
 \lefteqn{(1- A_1   (c_1\log_2 n)^{-1/2} 2^{c_1\log_2 n(h( \frac{1+c_2}{2})-1)}  )^{\frac{n}{c_1\log_2 n}-1} }\\ 
&=& (1- A_1   (c_1\log_2 n)^{-1/2} 2^{-(1-2\delta_2)\log_2 n} )^{\frac{n}{c_1\log_2 n}-1}\\
&= &  (1- A_1   (c_1\log_2 n)^{-1/2} n^{-(1-2\delta_2)} )^{\frac{n}{c_1\log_2 n} -1}     \\
&\leq & (1- (c_1\log_2 n) n^{-(1-\delta_2)} )^{\frac{n}{c_1\log_2 n}-1}  \\
&\leq & C e^{-n^{\delta_2}}(1- (c_1\log_2 n) n^{-(1-\delta_2)} )^{-1}\\
&\leq & e^{-\frac{1}{2}n^{\delta_2}}
\end{eqnarray*}
for $n$ large enough. 
This proves the claim made above.

We now consider the random variables $X_{k_j}$ for $k_j=[[c_1\log_2 n]/2]+j[c_1\log_2 n]$, $j=0,\ldots, [n/[c_1\log_2 n]]-1$ and 
bound the size of each weighted sum of the $\{a_k\}$ outside the interval containing $k_j$. 
That is, using Hoeffding's inequality again, for a given $k_j$, 
\begin{eqnarray}
\P\left(\left|  \sum_{l=[c_1\log_2 n/2]+1}^{ k_j}\frac{1}{l}a_{k_j-l+1}+\sum_{l=[c_1\log_2 n/2]+1}^{n-k_j}\frac{1}{l}a_{l+k_j-1}  \right|>t\right)\label{eachsum}\\
\leq 2 \exp(-[c_1 \log_2 n] \;t^2/2M^2).\nonumber
\end{eqnarray}
The probability that $t$ is exceeded for some $k_j$ is bounded by 
\begin{equation}
2n \exp(-[c_1 \log_2 n] \;t^2/2M^2).\nonumber
\end{equation}
By just setting $t=4M/\sqrt{c_1}$  
we have that each sum of the form inside~\eqref{eachsum} is bounded by $4M/\sqrt{c_1}$ with probability at least $1-\frac{1}{n}$. 

Let $E_0$ denote the event that both at least one subinterval of length $[c_1\log_2 n]$ satisfying the property 
discussed from~\cite{ER70} exists and that each sum outside this interval of the form inside~\eqref{eachsum} is bounded by $4M/\sqrt{c_1}$. 
This event occurs with probability at least $1-\frac{2}{n}$ for large $n$. 
That is, when $E_0$ occurs, there exists an interval where the number of $+1$'s is at least $ \frac{1+c_2}{2}[c_1\log_2 n]$ 
and the number of $-1$'s is at most $ \frac{1-c_2}{2}[c_1\log_2 n]$.
Denote by $\ks$ the $k_j$ corresponding to the interval with sufficiently many $+1$'s. 
Then,
\begin{eqnarray}
\lefteqn{\E \Big(\max_{1\leq k\leq n} |X_k|\Big)}\nonumber\\
&\geq & \E \Big(\max_{0\leq j< J-1} |X_{k_j}|\Big)\nonumber\\
&\geq& \E \Big(\max_{0\leq j< J-1} |X_{k_j}|\Big|E_0\Big)\P(E_0)\nonumber\\
&\geq&\Big(1-\frac{2}{n}\Big)\E \Big(\max_{0\leq j< J-1} |X_{k_j}|\Big|E_0\Big)\nonumber\\
&\geq& \Big(1-\frac{2}{n}\Big)\E \left(\sum_{l=1}^{[c_1\log_2 n/2]}\frac{1}{l}\epsilon_{\ks+l-1}|a_{\ks+l-1}|+\sum_{l=1}^{[c_1\log_2 n/2]}\frac{1}{l}\epsilon_{\ks-l}|a_{\ks-l}| \Bigg| E_0 \right)- \frac{2M}{\sqrt{c_1}}.\nonumber\\
&&\label{forlemma} 
\end{eqnarray}
We look at the expectation term in~\eqref{forlemma} and apply Lemma~\ref{conditionalsum}, which is given below. 
For each $r\geq  \frac{1+c_2}{2}[ c_1\log_2 n]$, all the subsets of a given interval of 
length $[c_1\log_2 n]$ with the number of $+1$'s equalling $r$ are equally probable. 
If we condition on a realization of the $\{|a_k|\}$, then the $|a_k|$ with the appropriate $\frac{1}{l}$ factors 
before them correspond to the $b$'s in the lemma. 
Since each $r\geq \frac{1+c_2}{2}[c_1\log_2 n ]$ a lower bound on the 
$k$ from Lemma~\ref{conditionalsum} is  $k=\frac{1+c_2}{2}[c_1\log_2 n]$ and
\begin{equation*}
 \frac{ 2\frac{1+c_2}{2} [c_1\log_2 n] -[c_1\log_2 n]}{[c_1\log_2 n]}=c_2.
\end{equation*}
This holds for any realization of the coefficients $\{a_k\}_{k=1}^\infty$. 
Therefore, 
\begin{eqnarray*}
\eqref{forlemma}&\geq & \Big(1-\frac{2}{n}\Big)c_2 \E \left(\sum_{l=1}^{[c_1\log_2 n/2]}\frac{1}{l}|a_{\ks+l-a}|+\sum_{l=1}^{[c_1\log_2 n/2]}\frac{1}{l}|a_{\ks-l}|   \right)-  \frac{2M}{\sqrt{c_1}}  \label{start}\\
&\geq &\Big(1-\frac{2}{n}\Big)c_2 \left(\sum_{l=1}^{[c_1\log_2 n/2]}\frac{m}{l}+\sum_{l=1}^{[c_1\log_2 n/2]}\frac{m}{l} \right)  -  \frac{2M}{\sqrt{c_1}} \\
&\geq & 2\Big(1-\frac{2}{n}\Big)c_2 m\log_2 \left(\frac{[c_1\log_2 n]}{2}\right)-\frac{2M}{\sqrt{c_1}}\\
&\geq &c_4m\log\log n\label{end}
\end{eqnarray*} 
for a constant $c_4$ for all large $n$. 
%&\geq & \E (\max_{1\leq k\leq K} |X_{[c_1\log n/2]+k[c_1\log n]|)
%\end{eqnarray} 
%\end{proposition}

\end{proof}

\begin{lemma}\label{conditionalsum}
Let $\epsilon\in \{\pm 1\}^p$ be uniformly distributed on 
\begin{equation*}
 \{x\in\{\pm 1\}^p:\; |\{i:\;x_i= 1\}|=k\}\nonumber
\end{equation*}
for some $1\leq k\leq p$ 
and let $b_1,\ldots,b_p$ be real numbers. 
Then
\begin{equation*}
 \E_{\epsilon} \sum_{i=1}^p \epsilon_i b_i= \frac{2k-p}{p}\sum_{i=1}^p b_i.
\end{equation*}

\begin{proof}
 Let $\mathcal{I}$ denote the set of subsets of $\{1,\ldots,p\}$ of cardinality $k$. 
The number of subsets $I\in\mathcal{I}$ such that $i\in I$ is equal to the number of subsets 
of $\{1,\ldots,p\}\backslash \{i\}$ of cardinality $k-1$. 
The cardinality of this set is $\binom{p-1}{k-1}=\frac{k}{p}\binom{p}{k}$. 
Therefore, 
\begin{eqnarray*}
 \E_\epsilon \Big(\sum_{i=1}^p \epsilon_i b_i\Big)&=& \binom{p}{k}^{-1} \sum_{I\in \mathcal{I}}( \sum_{i\in I}b_i -\sum_{i\notin I}b_i)\\
&=& \binom{p}{k}^{-1} (\sum_{I\in \mathcal{I}}\sum_{i\in I}b_i -\sum_{I\in \mathcal{I}}\sum_{i\notin I}b_i)\\
&=& \binom{p}{k}^{-1} (\sum_{i=1}^p b_i |I\in \mathcal{I}:\;i\in I|-\sum_{i=1}^p b_i |I\in \mathcal{I}:\;i\notin I| )\\
&=&\binom{p}{k}^{-1} \left(\frac{k}{p} \binom{p}{k} \sum_{i=1}^p b_i -\frac{p-k}{p} \binom{p}{k}\sum_{i=1}^p b_i\right)\\
&=& \frac{2k-p}{p}\sum_{i=1}^p b_i.
\end{eqnarray*}
\end{proof}

\end{lemma}

Using Proposition~\ref{propdiscrete}, we can now prove the main theorem.

\begin{proof}{\bf of Theorem~\ref{maintheorem}}
We start with the upper bounds for both the Gaussian and compact support cases. 
If $t<-n$, then
\begin{eqnarray*}
 \Big|\sum_{k=1}^n a_k\frac{\sin \pi (t-k)}{\pi (t-k)}\Big|&\leq & \max_{1\leq k\leq n} |a_k|\sum_{k=n+1}^{2n}\frac{1}{\pi k}\\
&\leq & \max_{1\leq k\leq n} |a_k|,
\end{eqnarray*}
and the same argument holds for $t>2n$. 
Therefore we bound the expectation of the supremum over $t\in [-n,2n]$, which will always be at least the order of $\max_{1\leq k\leq n} |a_k|$. 
Throughout we use $|\sin x|\leq |x|$. 
We have
\begin{eqnarray*}
\sup_{t\in [-n,2n]}\Big|\sum_{k=1}^n a_k\frac{\sin \pi (t-k)}{\pi (t-k)}\Big|\label{first}&\leq & \max_{-n\leq l\leq 2n-1}\;\sup_{t\in [0,1]}\Big|\sum_{k=1}^n a_k\frac{\sin \pi (l+t-k)}{\pi (l+t-k)}\Big|\\
&=&  \max_{-n\leq l\leq 2n-1}\;\sup_{t\in [0,1]}\Big|\sum_{k=1}^n a_k\frac{(-1)^{l-k}\sin \pi t}{\pi (l+t-k)}\Big|.\label{expression}
\end{eqnarray*}

We use
\begin{eqnarray}
\lefteqn{\sup_{t\in [0,1]}\max_{-n\leq l\leq 2n-1}\Big|\sum_{k=1}^n a_k\frac{(-1)^{l-k}\sin \pi t}{\pi (l+t-k)}\Big|}\label{beginning}\nonumber\\
&\leq & \sup_{|t|\leq \frac{1}{n}}\max_{-n\leq l\leq 2n}\Big|\sum_{k=1}^n a_k\frac{(-1)^{l-k}\sin \pi t}{\pi (l+t-k)}\Big|\label{firstterm}\\
&&+\sup_{\frac{1}{n}\leq |t|\leq \frac{1}{2}}\max_{-n\leq l\leq 2n}\Big|\sum_{k=1}^n a_k\frac{(-1)^{l-k}\sin \pi t}{\pi (l+t-k)}\Big|.\label{secondterm}
%&\leq& \max_{l=1,\ldots,n^2}|\sum_{k=1}^n a_k\frac{\sin \pi (\frac{l}{n^2}-k)}{\pi (\frac{l}{n^2}-k)}|+\frac{1}{n}\max_{t\in [0,n]}|\frac{d}{dt}\sum_{k=1}^n a_k\frac{\sin \pi (t-k)}{\pi (t-k)}|\\
%&= & \max_{l=1,\ldots,n^2}|\sum_{k=1}^n a_k\frac{\sin \pi (\frac{l}{n^2}-k)}{\pi (\frac{l}{n^2}-k)}|+\frac{1}{n}\max_{t\in [0,n]}|\sum_{k=1}^n a_k \frac{-\pi(t-k)\cos \pi (t-k)-\pi \sin\pi (t-k)}{\pi^2(t-k)^2}|\\
\end{eqnarray}
For the term~\eqref{firstterm}, we choose an arbitrary $-n\leq l_0\leq 2n$ and obtain
\begin{eqnarray}
\sup_{|t|\leq \frac{1}{n}}\Big|\sum_{k=1}^n a_k\frac{(-1)^{l_0-k}\sin \pi t}{\pi (l_0+t-k)}\Big|&\leq&\max_{1\leq k\leq n} |a_k|\sup_{|t|\leq \frac{1}{n}}\sum_{k=1}^n\frac{ |\sin \pi t|}{\pi|l_0+t-k|}\nonumber\\
&\leq & \max_{1\leq k\leq n} |a_k| \sup_{|t|\leq \frac{1}{n}}\left(\frac{|\sin \pi t|}{\pi |t|}+\sum_{k=1,k\neq l_0}^n\frac{ |\sin \pi t|}{\pi |l_0+t-k|}\right)\nonumber\nonumber\\
&\leq & \max_{1\leq k\leq n} |a_k| \sup_{|t|\leq \frac{1}{n}} \left(1+\frac{1}{n}\sum_{k=1,k\neq l_0}^n\frac{ 1}{|l_0+t-k|}\right)\nonumber\\
&\leq & \max_{1\leq k\leq n} |a_k| \left(1+\frac{2}{n}\sum_{k=1}^n\frac{ 1}{k-\frac{1}{n}}\right)\nonumber\\
&\leq & \max_{1\leq k\leq n} |a_k| \left(1+\frac{2}{n}\left(\frac{n}{n-1}+1+\log n \right)\right)\nonumber\\
&\leq &2 \max_{1\leq k\leq n} |a_k|.\label{term1}
\end{eqnarray}
Since this holds for each $l$ we have bounded~\eqref{firstterm} by~\eqref{term1}.
%Now we define $\Tn$ by
%\begin{equation}
%\Tn=\{x\in [-n,2n]:\; |x-k|\geq \frac{1}{n}\;\;\textnormal{for all }k\in\N\}.
%\end{equation}
Now we look at~\eqref{secondterm}, 
and use that if $f$ is differentiable on the interval $[a,b]$, then 
\begin{equation*}
\sup_{t\in [a,b]}|f(t)|\leq \max\{|f(a)|,|f(b)|\}+|b-a|\cdot\sup_{t\in [a,b]}\Big|\Big(\frac{d}{dt} f\Big)(t)\Big|.
\end{equation*}
We then have 
\begin{eqnarray*}
\lefteqn{ \max_{-n\leq l\leq 2n-1} \sup_{\frac{1}{n}\leq |t|\leq \frac{1}{2}}\Big| \sum_{k=1}^n a_k\frac{(-1)^{l-k}\sin \pi t}{\pi (l+t-k)}\Big|}\\
& \leq & \max_{l=-n,\ldots,2n,\; r=1,\ldots,n-1} \Big|\sum_{k=1}^n a_k\frac{(-1)^{l-k}\sin \pi \frac{r}{n}}{\pi (l+\frac{r}{n}-k)}\Big|\\
&&+  \max_{-n\leq l\leq 2n} \sup_{\frac{1}{n}\leq |t|\leq \frac{1}{2}} \frac{1}{n}\Big|\frac{d}{dt}\sum_{k=1}^n a_k\frac{(-1)^{l-k}\sin \pi t}{\pi (l+t-k)}\Big|.
\end{eqnarray*}
%Now choose an arbitrary $l_0$ such that $-n\leq l_0\leq 2n$. 
For the second term we have the bound
Then 
\begin{eqnarray}
\lefteqn{ \max_{-n\leq l\leq 2n-1} \sup_{\frac{1}{n}\leq |t|\leq \frac{1}{2}} \frac{1}{n}\Big|\frac{d}{dt}\sum_{k=1}^n a_k\frac{(-1)^{l-k}\sin \pi t}{\pi (l+t-k)}\Big|}\nonumber\\
&=& \max_{-n\leq l\leq 2n-1} \sup_{\frac{1}{n}\leq |t|\leq \frac{1}{2}} \frac{1}{n}\Big|\sum_{k=1}^n a_k(-1)^{l-k}\frac{\pi^2(l+t-k)\cos \pi t-\pi \sin\pi t }{\pi^2 (l+t-k)^2}\Big|\nonumber\\
&\leq & \frac{1}{n}\max_{1\leq k\leq n}|a_k| \max_{-n\leq l\leq 2n-1}\sup_{\frac{1}{n}\leq |t|\leq \frac{1}{2}} \left( \sum_{k=1}^n\frac{1}{|l+t-k|}+\sum_{k=1}^n\frac{|\sin \pi t|}{\pi(l+t-k)^2}\right) \nonumber\\
&\leq & \frac{1}{n}\max_{1\leq k\leq n}|a_k| \max_{-n\leq l\leq 2n-1}\sup_{\frac{1}{n}\leq |t|\leq \frac{1}{2}} \left(\frac{1}{t}+ \sum_{k=1,k\neq l}^n\frac{1}{|l+t-k|}+\frac{|\sin \pi t|}{\pi t^2}+\sum_{k=1,k\neq l}^n\frac{1}{\pi(l+t-k)^2}\right)\nonumber \\
&\leq &\frac{1}{n}\max_{1\leq k\leq n}|a_k|\left( 2n+2\log n+2\right)\label{lessthanone}\nonumber\\
&\leq & 4\max_{1\leq k\leq n}|a_k|\label{term2}
\end{eqnarray}
for sufficiently large $n$.  
Thus, in both the Gaussian and compact support cases we have to find a bound on
\begin{equation}
\max_{l=-n,\ldots,2n;\;r=1,\ldots,n-1}\Big|\sum_{k=1}^n a_k\frac{(-1)^{l-k}\sin \pi \frac{r}{n}}{\pi (l+\frac{r}{n}-k)}\Big|.\label{tobound}
\end{equation}

We start with the Gaussian case. 
For $l=-n,\ldots,2n$, and $r=1,\ldots,n-1$ set 
\begin{equation*}
\xi_{l,r}= \sum_{k=1}^n a_k\frac{\sin \pi \frac{r}{n}}{\pi (l+\frac{r}{n}-k)}.
\end{equation*}
The $\{\xi_{l,r}\}$ are Gaussian random variables, so that  
using inequality~(3.13) in~\cite{LT91}, 
\begin{eqnarray}
 \E \Big(\max_{-n\leq l\leq 2n,1\leq r\leq n-1}|\xi_{l,r}|\Big)&\leq &3\sigma\sqrt{\log 3n^2}+\max_{-n\leq l\leq  2n,1\leq r\leq n-1}(\E |\xi_{l,r}|^2)^{1/2}\nonumber\\
&\leq & C_1\sigma\sqrt{\log n}.\label{fromLT}
\end{eqnarray}
We return to~\eqref{beginning}, and collecting the terms in~\eqref{term1}, \eqref{term2} and~\eqref{fromLT} we have 
\begin{eqnarray*}
\E \Big(\max_{t\in [-n,2n]}|\sum_{k=1}^n a_k\frac{(-1)^{l-k}\sin \pi (t-k)}{\pi (t-k)}|\Big)&\leq & 6\E \Big(\max_{1\leq k\leq n}|a_k|\Big)+C_1\sigma \sqrt{\log 3 n^2}\nonumber\\
&\leq &C\sigma\sqrt{\log n},
\end{eqnarray*}
where we have applied   inequality~(3.13) in~\cite{LT91} to $\E \big(\max_{1\leq k\leq n}|a_k|\big)$ as well.

Now we address the case when the $\{a_k\}_{k=1}^\infty$ are symmetrically distributed and satisfy $\P(|a_k|>M)=0$ for all $k$. 
First we replace~\eqref{tobound} for an expression with the same distribtuion:
\begin{equation*}
\eqref{tobound}\stackrel{\textnormal{dist}}{=}\max_{l=-n,\ldots,2n;\;r=1,\ldots,n-1}\Big|\sum_{k=1}^n a_k\frac{\sin \pi \frac{r}{n}}{\pi |l+\frac{r}{n}-k|}\Big|.
\end{equation*}
For $r$ satisfying $\frac{r}{n}\leq \frac{1}{2}$, 
we have 
\begin{equation*}
\max_{l=-n,\ldots,2n}\left|\sum_{k=1}^n a_k\frac{\sin \pi \frac{r}{n}}{\pi |l+\frac{r}{n}-k|}\right|
\leq M+\max_{l=-n,\ldots,2n}\left|\sum_{k=1,k\neq l}^n \frac{a_k}{|l+\frac{r}{n}-k|}\right|,
\end{equation*}
and if $\frac{1}{2}<\frac{r}{n}$ we remove the term indexed by $k=l+1$ to obtain a similar expression. 
We then have 
\begin{equation*}
\left|\sum_{k=1,k\neq l}^n \frac{a_k}{l+\frac{r}{n}-k}\right|
=\left|\sum_{j=2}^l \frac{a_{l-j+1}}{j-1+\frac{r}{n}}  +\sum_{j=2}^{n-l+1} \frac{a_{l+j-1}}{1+\frac{r}{n}-j}\right|
\end{equation*}
when $\frac{r}{n}\leq \frac{1}{2}$ and an anologous term when $\frac{r}{n}>\frac{1}{2}$. 
Now we need to slightly adjust the argument given in the proof of Proposition~\ref{propdiscrete}. 
Similar to equations~\eqref{defY} and~\eqref{defZ},
for $-n\leq l\leq 2n$ and $1\leq r\leq n-1$ we define 
\begin{equation*}
Y_{l,r}= \sum_{j=2}^l \frac{a_{l-j+1}}{j-1+\frac{r}{n}}   \label{defYa}
%\;\;\;\textnormal{and}\;\;\;Y_{l,r}=\sum_{l=2}^k \frac{a_{k-l+1}}{l+\frac{r}{n}}  \;\;\textnormal{ if }\;\;\frac{r}{n}> \frac{1}{2} 
\end{equation*}
and 
\begin{equation*}
 Z_{l,r}= \sum_{j=2}^{n-l+1} \frac{a_{l+j-1}}{1+\frac{r}{n}-j} \label{defZa}
%\;\;\;\textnormal{and}\;\;\;Z_{l,r}=  \sum_{l=2}^{n-k+1} \frac{a_{l+k-1}}{l-\frac{r}{n}}  \;\;\textnormal{ if }\;\;\frac{r}{n}> \frac{1}{2}    .
\end{equation*}
for $\frac{r}{n}\leq \frac{1}{2}$ 
and the anologous terms when $\frac{r}{n}> \frac{1}{2}$. 
We can now pursue bounds for the 
$Y$ and the $Z$ terms individually, as we did earlier. 
Just as in~\eqref{useL}, for any $L$, if $l\leq L$ and $\frac{r}{n}\leq \frac{1}{2}$ we have
\begin{equation*}
|Y_{l,r}|\leq M\Big(1+\sum_{j=1}^{l}\frac{1}{j-\frac{r}{n}}\Big)\leq M+M\log(l)\leq M(1+\log L).
\end{equation*}
We define $\YL_{l,r}$ analogously to~\eqref{defYL} and note that its variance is bounded by $M^2/(L+1)$, 
so that the same type of exponential bound applies as in Proposition~\ref{propdiscrete}. 
Where we had a maximum over $n-L$ random 
variables in~\eqref{union-n} and used the union bound, 
we now have $3n^2-L$ and again use a union bound, resulting in a factor $3n^2$ where we had $n$ in~\eqref{factor-n}. 
To take care of the $n^2$, we set $\delta=M\frac{2\sqrt{\log n}}{\sqrt{L}}$, rather than $\delta=M\frac{\sqrt{2\log n}}{\sqrt{L}}$ 
as it was earlier. 
Then the same argument as was made in the proof of Proposition~\ref{propdiscrete} applies here, 
thus giving the upper bound of $C_2\log \log n$. 

We now prove the lower bounds and start with the Gaussian case. 
Let $a_k\sim \mathcal{N}(0,\sigma^2)$ for a fixed $\sigma^2$ for all $k$. 
Then, considering $t$ at the integers, 
\begin{eqnarray*}
\E\left(\sup_{t\in \R}\left|\sum_{k=1}^n a_k \frac{\sin \pi (t-k)}{\pi (t-k)}\right|\right)&\geq & \E\left(\max_{t\in \{1,2,\ldots,n\}}\left|\sum_{k=1}^n a_k \frac{\sin \pi (t-k)}{\pi (t-k)}\right|\right)\\
&=&\E\Big( \max_{1\leq k\leq n}|a_k|\Big). 
\end{eqnarray*}
%For a fixed $t$ we define 
%\begin{equation*}
%\Yt_l =\sum_{k=1}^n a_k\frac{(-1)^{l-k}\sin \pi t}{\pi (l+t-k)}\;\;\textnormal{  for  }l=1,\ldots ,n.
%\end{equation*}
%For $l\neq m$, and assuming without loss of generality that $m>l$,  
%\begin{eqnarray*}
%\E\big(| \Yt_l -\Yt_m|^2\big)&=& \E \left(\Big|\sum_{k=1}^n a_k\left(\frac{(-1)^{l-k}\sin \pi t}{\pi (l+t-k)}-\frac{(-1)^{m-k}\sin \pi t}{\pi (m+t-k)}\right)\Big|^2\right)\\
%&=& \sigma^2  \sum_{k=1}^n \left|\frac{(-1)^{l-k}\sin \pi t}{\pi (l+t-k)}-\frac{(-1)^{m-k}\sin \pi t}{\pi (m+t-k)}\right|^2\\
%&\geq & \sigma^2\Big|\frac{\sin\pi t}{\pi t}-\frac{(-1)^{m-l}\sin \pi t}{\pi (m-l+t)}\Big|^2\\
%&\geq & \sigma^2\Big|\frac{\sin\pi t}{\pi t}-\frac{\sin \pi t}{\pi (m-l+t)}\Big|^2\\
%&\geq & \sigma^2\Big|\frac{\sin\pi t}{\pi t}-\frac{\sin \pi t}{\pi (1+t)}\Big|^2,
%\end{eqnarray*}
%so that 
%\begin{equation*}
%\E \big(| Y^{(1/2)}_l-Y^{(1/2)}_m|^2\big)\geq \sigma^2 \Big(\frac{4}{3\pi}\Big)^2.
%\end{equation*}
The lower bound follows from the standard fact that the expected maximum of $n$ 
independent 
Gaussian random variables with variance $\sigma^2$ is of order $\sigma \sqrt{\log n}$. 
%\begin{equation*}
%\E \Big(\max_{1\leq k\leq n}|a_k |\Big)\geq c_1\sigma \sqrt{\log n}.
%\end{equation*}
%This gives a lower bound for the Gaussian case.  

Lastly, we show the lower bound for the random variables with bounded support. 
We consider $t\in \{\frac{1}{2}, 2+\frac{1}{2},4+\frac{1}{2},\ldots,2[\frac{n-1}{2}]+\frac{1}{2}\}$. 
Then 
\begin{eqnarray}
 \sup_{t\in \R} \Big|\sum_{k=1}^na_k\frac{\sin \pi (t-k)}{\pi(t-k)}\Big|&\geq &\max_{0\leq l\leq [\frac{n-1}{2}]}\Big|\sum_{k=1}^na_k\frac{\sin \pi (2l+\frac{1}{2}-k)}{\pi(2l+\frac{1}{2}-k)}\Big|\nonumber\\
&=& \max_{0\leq l\leq [\frac{n-1}{2}]}\Big|\sum_{k=1}^na_k\frac{(-1)^{2l-k}}{\pi(2l+\frac{1}{2}-k)}\Big|\nonumber\\
&=&  \max_{0\leq l\leq [\frac{n-1}{2}]}\Big|\sum_{k=1}^n (-1)^{k}a_k\frac{1}{\pi(2l+\frac{1}{2}-k)}\Big|\nonumber\label{withminus}\\
&\stackrel{\textnormal{dist}}{=}& \max_{0\leq l\leq [\frac{n-1}{2}]}\Big|\sum_{k=1}^n a_k\frac{1}{\pi(2l+\frac{1}{2}-k)}\Big|\label{withoutminus},
\end{eqnarray}
where  equality~\eqref{withoutminus} 
is in distribution and 
holds due to the assumption that the $\{a_k\}_{k=1}^\infty$ are symmetrically distributed. 
The lower bound on the expectation of~\eqref{withoutminus} is proved analogously to the lower bound in Proposition~\ref{propdiscrete}.
\end{proof}

\vspace{1cm}

We close with a question. 
In the communications setting discussed in the Introduction, one is interested in methods to reduce the peak value of a signal. 
While the results presented here show that  peak value grows more mildly for shifted sinc kernels than for 
orthonormal functions on the unit interval, 
how the peak can be reduced and the limits to certain methods for doing so are still of interest. 
One such method is to allow a subset of the coefficients to be random, and choose the remaining coefficients to reduce the peak of the signal 
resulting from the random coefficients. 
Limiting behavior for this scheme in the Fourier setting on the unit interval was addressed in~\cite{BF11}. 
However, it is not readily apparent what the analogous behavior for shifted sinc kernels is. 

\def\cprime{$'$} \def\cprime{$'$} \def\cprime{$'$} \def\cprime{$'$}
\bibliographystyle{abbrv}

%\bibliographystyle{unsrt}
%\bibliography{../Bibfiles/Banach,../Bibfiles/math,../Bibfiles/cs,../Bibfiles/ee,../Bibfiles/mathbook}

\begin{thebibliography}{9}

\bibitem{BF11}
H.~Boche and B.~Farrell.
\newblock {PAPR and the Density of Information Bearing Signals in OFDM}.
\newblock {\em EURASIP Journal on Advances in Signal Processing}.
\newblock vol. 2011, Article ID 561356, 9 pages, 2011.



\bibitem{BF11a}
H.~Boche and B.~Farrell.
\newblock {On the Peak-to-Average Power Ratio Reduction Problem for Orthonormal
  Transmission Schemes}.
\newblock 2010.
\newblock Submitted.

\bibitem{ER70}
P.~Erd{\H{o}}s and A.~R{\'e}nyi.
\newblock On a new law of large numbers.
\newblock {\em J. Analyse Math.}, 23:103--111, 1970.


\bibitem{Hoe63}
W.~Hoeffding.
\newblock Probability inequalities for sums of bounded random variables.
\newblock {\em J. Amer. Statist. Assoc.}, 58:13--30, 1963.


\bibitem{KT95}
B.~Kashin and L.~Tzafriri.
\newblock Lower estimates for the supremum of some random processes.
\newblock {\em East J. Approx.}, 1(1):125--139, 1995.


\bibitem{LT91}
M.~Ledoux and M.~Talagrand.
\newblock {\em Probability in {B}anach {S}paces}.
\newblock Springer-Verlag, Berlin, 1991.

\bibitem{Lit07}
S.~Litsyn.
\newblock {\em Peak Power Control in Multicarrier Communications}.
\newblock Cambridge Univ. Press, New York, 2007.


%\bibitem{Tal05}
%M.~Talagrand.
%\newblock {\em The {G}eneric {C}haining}.
%\newblock Springer Monographs in Mathematics. Springer-Verlag, Berlin, 2005.

\end{thebibliography}

\end{document}